\documentclass[graybox]{svmult}

\usepackage{mathptmx}       
\usepackage{helvet}         
\usepackage{courier}        
\usepackage{makeidx}         
\usepackage{graphicx}        
\usepackage{multicol}        
\usepackage[bottom]{footmisc}

\makeindex             

\usepackage{amsmath}
\usepackage{amsfonts}

\newcommand{\F}{\mathbb{F}}

\newcommand{\rowsp}{\textnormal{rowsp}}

\smartqed

\begin{document}

\title*{Subspace Fuzzy Vault}
\author{Kyle Marshall, Davide Schipani, Anna-Lena Trautmann and Joachim Rosenthal}
\authorrunning{K. Marshall, D. Schipani, A.-L. Trautmann and J. Rosenthal} 
\institute{Kyle Marshall,  Davide Schipani, Joachim Rosenthal \at Institute of Mathematics, University of Zurich, Switzerland\\ 
\email{\{kyle.marshall, davide.schipani, joachim.rosenthal\}@math.uzh.ch}\\
KM and JR were supported by Swiss National Science Foundation Grant no.\ 149716.
\and Anna-Lena Trautmann \at Department of Electrical and Electronic Engineering, University of Melbourne, Australia, and \\Department of Electrical and Computer Systems Engineering, Monash University, Australia \\ \email{anna-lena.trautmann@unimelb.edu.au}\\
ALT was supported by Swiss National Science Foundation Fellowship no.\ 147304. }
%
%
\maketitle

\abstract{
	Fuzzy vault is a scheme providing secure
	authentication based on fuzzy matching of sets. A major application is
	the use of biometric features for authentication, whereby unencrypted
	storage of these features is not an option because of security
	concerns.  While there is still ongoing research around the practical
	implementation of such schemes, we propose and analyze here an
	alternative construction based on subspace codes. This offers
	some advantages in terms of security, as an eventual discovery of the
	key does not provide an obvious access to the features. Crucial for
	an efficient implementation are the computational complexity and the
	choice of good code parameters. The parameters depend on the
	particular application, e.g.\ the biometric feature to be stored and the
	rate one wants to allow for false acceptance.  The developed theory is
	closely linked to constructions of subspace codes studied in the area
	of random network coding.
\vspace{1cm}
}

\section{Introduction}

Fuzzy vault is the term used by Juels and Sudan in \cite{Juels:2006:FVS:1110940.1110956} to describe a cryptographic primitive, in which a key $\kappa$ is hidden by a set of features $A$ in such a way, that any witness $B$, which is close enough to $A$ under the set difference metric, can decommit $\kappa$.  Fuzzy vault is related to the fuzzy commitment scheme of Juels and Wattenberg \cite{Juels:1999:FCS:319709.319714}, which gives a solution for noisy hashing of data for the Hamming distance. This and a dual version of it, the fuzzy syndrome hashing scheme, were considered by the authors in \cite{BBCRS, FMRST,sc10p}.

The motivation for fuzzy vault is related to the growing interest in using fuzzy authentication systems, i.e.\ systems that do not require an exact match, but rather a partial one, between two sets. Instances include the use of biometric features for authentication, personal entropy systems to allow password recovery by answering a set of questions with a level of accuracy above a certain threshold, privacy-protected matching to allow find a match between two parties without disclosing the features in public.       

In early biometric authentication systems, comparison of a biometric was done against an image stored locally on the machine, rather than in some hashed form. For security purposes however, passwords are normally stored in hashed form. Moreover, since biometric data is irreplaceable in the sense that once compromised it cannot be changed, storing the data in un-hashed form can pose a significant security risk \cite{Clancy03securesmartcard-based}. Biometric data is inherently noisy, however, so direct hashing of a user's features would prevent the authentic user from accessing the system, as no error tolerance in the matching would be allowed. Using error correcting techniques, the fuzzy vault is a scheme that can recover a secret key hidden by features even in the presence of noise. Recent advancements have been made in the pre-alignment of biometrics (cf. \cite{Li:2010:AFC:1786801.1786970} and references therein), specifically fingerprints, allowing for comparative methods without storage of the image itself. These advancements make fuzzy vault a promising and feasible cryptographic solution for noisy data.

Recently, much work has been done in the area of error correcting codes in projective space. These codes turn out to be appropriate for error correction in random network coding  \cite{Koetter07codingfor}, and are referred to as error correcting random network codes, projective space codes, or subspace codes. The aim of this paper is to show that the construction of the fuzzy vault in \cite{Juels:2006:FVS:1110940.1110956} can be extended and adapted to work for subspace codes in an analogous way with advantages and limitations. Namely, we present a construction for a fuzzy vault based on constant dimension subspace codes, a class of error correcting codes in projective $n$-space over a finite field $\F_q$. For illustration, an example will be provided by using spread codes, a particular class of subspace codes. 

The rest of the paper is organized as follows: Section \ref{sec:background} provides preliminaries, terminology and refreshes the original fuzzy vault scheme. Section \ref{s:SFV} presents the new scheme based on subspace codes. Section \ref{s:secu} relates to security and examples and lastly Sections \ref{sec:considerations} and  \ref{sec:conclusions} give further considerations and concluding final remarks.

\section{Preliminaries}\label{sec:background}

Denote by $\F_q$ the finite field with $q$ elements, where $q$ is a prime power. The  \emph{set difference metric}
$d_{\Delta}$ is defined as
\[d_{\Delta}(A,B) := |(A\backslash B) \cup (B\backslash A) |    , \quad A,B \subseteq \F_q \]
and the \emph{Hamming metric} $d_H$ is defined as
\[d_H(u,v) := |\{ i \mid u_i \neq v_i \}|    , \quad u=(u_1,\dots,u_n),v=(v_1,\dots,v_n)\in \F_q^n .\]
Let $g_1,\dots, g_n\in \F_q^*$ be distinct elements. A $k$-dimensional \emph{Reed-Solomon code} $\mathcal C\subseteq \F_q^n$ can be defined as 
\[\mathcal C=\{(f(g_1),\dots, f(g_n) ) \mid  f(x) \in \F_q[x], \deg(f) < k  \}.\]
It has minimum Hamming distance $d_{\min,H} (\mathcal C)=n-k+1$ and cardinality $|\mathcal C|=q^{k}$ \cite{ma77}.

A \emph{constant dimension (subspace) code} is a subset of the Grassmannian $\mathcal{G}_q(k,n)$, the set of all $k$-dimensional subspaces of $\F_q^n$. The subspace distance defines a metric on $\mathcal{G}_q(k,n)$, given by 
$$d_S(U,V) := \dim(U+V)-\dim(U\cap V), \quad U,V \in \mathcal{G}_q(k,n)$$
 for $U,V\in \mathcal{G}_q(k,n)$ \cite{Koetter07codingfor}. While finding good subspace codes is still an open research problem, there are many candidates now, including the Reed-Solomon-like and spread code constructions \cite{Koetter07codingfor, DBLP:conf/isit/ManganielloGR08}.
An explicit construction of a spread code can be found in \cite{DBLP:conf/isit/ManganielloGR08}, and it is this construction we use as the definition of a spread code: 
Let $p(x) \in \F_q[x]$ be an irreducible monic polynomial of degree $k$ and $P\in \F_q^{k\times k}$ be its companion matrix. Let $n = ks$ for $s \in \mathbb{N}$. Then, 
$$\mathcal{S} = \{\rowsp(A_1 \mid \cdots \mid A_s)\mid A_i\in \F_q[P] , (A_1 \mid \cdots \mid A_s) \neq (0 \mid \cdots \mid 0)\}$$ 
is called a $(k,n)$-\emph{spread code}, where $\rowsp(A)$ is the row space of a matrix $A$. From the definition, one can see that the minimum subspace distance of a spread code is $d_{\min,S}(\mathcal{ S}) = 2k$ and that the cardinality is $|\mathcal{S}| = \frac{q^n-1}{q^k-1}$. 

For practical purposes we need a unique representation of subspaces, and we will choose their matrix representation in reduced row echelon form (i.e. the matrix  in reduced row echelon form whose row space is the respective subspace) as such. 




We will now briefly revisit the fuzzy vault scheme \cite{Juels:2006:FVS:1110940.1110956}. We will refer to the following description (cf. also \cite{Hartloff}), although we are aware of different interpretations of the scheme throughout the literature, especially in terms of the decoding algorithms and parameters (\cite{Poon_efficient}). Since this scheme  is based on polynomial evaluation, it will henceforth be called the \emph{polynomial fuzzy vault} (PFV) scheme. 

Let $\kappa = (k_0,k_1,...,k_{\ell-1})\in \F_{q}^\ell$ be the secret key and $\kappa(x) = k_0+k_1x + ... k_{\ell-1}x^{\ell-1}\in \F_q[x]$ the corresponding key polynomial. 
 Let $A\subset \F_{q}\backslash \{0\}$ be the set of genuine features with $|A| = t> \ell$. Furthermore, let $\lambda: \F_{q} \rightarrow \F_{q}$ be a random map  such that $\lambda(x)\neq \kappa(x)$ for all $x\in B$. 
Choose $r > t$ and select a set $B\subset \F_{q}\backslash A$ such that  $|B|=r-t$. Construct  the sets 
\begin{align*} 
\mathcal{P}_{auth} &= \{(x,\kappa(x))\mid x\in A\}, \\ 
\mathcal{P}_{chaff} &= \{(x, \lambda(x))\mid x\in B\}, \\
\mathcal{V} &= \mathcal{P}_{auth}\cup \mathcal{P}_{chaff} .
\end{align*} 
We will call $\mathcal{P}_{auth}$ the set of \emph{authentic points},  $\mathcal{P}_{chaff} $ the set of \emph{chaff points} and $\mathcal{V} $ the set of \emph{vault points}. 

The remaining parts of the fuzzy vault are a code and a corresponding error correcting decoding algorithm. The code is the $\ell$-dimensional Reed-Solomon code $\mathcal{C} \subseteq\F_q^{t}$, 
\[\mathcal C=\{(f(g_1),\dots, f(g_t) ) \mid  f(x) \in \F_q[x], \deg(f) < \ell  \},\]
whose defining distinct evaluation points $g_{1},\dots, g_{t}$ are the points in $A$, i.e.\ the genuine features. The key polynomial $\kappa(x)$ gives rise to a codeword of $\mathcal{C}$. 
If a witness attempts to gain access to the key, the witness submits a set  of features $W\subset \F_{q}$. Let $Z\subseteq \mathcal{V}$ be the set of vault points $(x,y)$ with $x\in W$. As the error correction capability of $\mathcal{C}$ is $\lfloor(t-\ell)/2\rfloor$, the witness needs $|Z\cap \mathcal{P}_{auth}|\geq t-\lfloor(t-\ell)/2\rfloor = \lceil(t+\ell)/2\rceil$ to recover $\kappa(x)$ with the decoding algorithm.

To simplify the setting and have a more workable model, assume that $|W|=t$ and that $B=\F_{q}\backslash A$. Then $|Z|=t$ and we can rewrite $d_{\Delta}(A,W)=2t - 2|A\cap W|= 2t-2|Z\cap \mathcal{P}_{auth}|)$. Thus the witness gains access to the key if
$$ d_{\Delta}(A,W)\leq 2t- (t+\ell) \iff d_{\Delta}(A,W)\leq d_{min,H}(\mathcal C)-1.$$

It was shown in \cite{DBLP:conf/biosig/MihailescuMT09} that certain reasonable parameters for the PFV scheme cause the system to be susceptible to a brute force attack. Choi et al.\ in \cite{YongChoi2008725} speed up the attack by using a fast polynomial reconstruction algorithm. These attacks may indicate that additional security measures should be taken to prevent the loss of a user's features.
A different type of security analysis is provided in \cite{Hartloff}.




\section{A Fuzzy Vault Scheme Utilizing Subspace Codes}\label{s:SFV}

We will now explain our new variant of the fuzzy vault scheme, and call this particular implementation the \emph{subspace fuzzy vault} (SFV) scheme. Unlike the PFV scheme in which the key is given by the coefficients of a polynomial, the key $\hat \kappa$ in this scheme is a subspace with a disguised generator matrix $\kappa$ (not in reduced row echelon form). 


\begin{definition}\label{def:SFV}
Let $k\leq n$, $\mathcal{C}\subset \mathcal{G}_q(k,n)$ a constant dimension subspace code, and $\hat{\kappa}\in\mathcal{C}$ a secret subspace. Choose some $\kappa \in \F_q^{k\times n}$ such that $\rowsp(\kappa) = \hat\kappa$. We will hide the key by a set of linearly independent features $A\subset \F_q^k$ with $|A|= k$ and a set $B= \F_q^k\backslash A$. Let $\lambda(x):\F_{q}^{k}\rightarrow \F_{q}^{n}$ be a random map such that $\lambda(x)\not \in \rowsp(\kappa)$ for all $x\in B$. Define the sets
\begin{align*} 
\mathcal{P}_{auth} &= \{(x,x \kappa)\mid x\in A\}, \\ 
\mathcal{P}_{chaff} &= \{(x, \lambda(x))\mid x\in B\}, \\
\mathcal{V} &= \mathcal{P}_{auth}\cup \mathcal{P}_{chaff}.
\end{align*} 
$\mathcal{P}_{auth}$ is called the set of \emph{authentic points}, $\mathcal{P}_{chaff}$ is called the set of \emph{chaff points}, and $\mathcal{V} $ the set of \emph{vault points}. 
\end{definition}







In order for a witness to decommit $\hat{\kappa}$, a set $W \subset \F_{q}^k$ is submitted and the second coordinates of the elements in the vault whose first coordinates correspond to $W$ are used to generate a subspace $W'$. 
This subspace is then decoded to yield a codeword $U\in \mathcal C$. We assume that $W$ consists of at most $k$ linearly independent features.

For a set $S\subset \F_q^k$, we will denote by $\langle S\rangle_\kappa$ the subspace spanned by the elements $\{s\kappa\mid s\in S\}$. 
We will also assume $\dim(W')=|W|$, although this may not happen, introducing some probability of error, as we mention below. The assumption is justified by estimating its probability using counting formulas like that in the following Lemma \ref{l:spanspace}, whilst supposing $n$ big enough and the second coordinates of the chaff points being randomly chosen within their domain.



\begin{theorem}\label{eq:sfvdistance} 
In the setting of Definition \ref{def:SFV}, the vault recovers the key $\hat{\kappa}$ if  and only if
\begin{equation*}
d_{\Delta}(A,W) \leq \frac{1}{2}(d_{\min,S}(\mathcal{C})- 1).
\end{equation*} 
\end{theorem}
\begin{proof}
We can express $W' = (W' \cap \hat{\kappa}) \oplus E$ for some subspace $E\subset \F_q^n$. 
As shown in \cite{Koetter07codingfor},  we can uniquely recover $\hat{\kappa}$ from $W'$ if and only if $ d_S(W', \hat{\kappa}) \leq \frac{1}{2}(d_{\min,S}(\mathcal{C}) - 1) $. 
%

%
Using properties of the rank and linear algebra identities, we get
\begin{align*} 
d_{\Delta}(A,W) &= |W\setminus A| + |A\setminus W|\\
 &=  \dim(\langle W\setminus A\rangle_\kappa) + \dim(\langle A\setminus W\rangle_\kappa) \\ 
&= \dim(\langle W\setminus A\rangle_\kappa) + k-\dim(\langle A\cap W\rangle_\kappa) \\ 
&= \dim(E) + k- \dim(\hat{\kappa}\cap W' ) \\
&= d_S(W', \hat{\kappa}).
\end{align*} 
Indeed, as $|W|\leq k$, $|A|= k$, and $W$ and $A$ are sets of linearly independent features, Sylvester's rank inequality implies $|W\setminus A|\leq\dim(\langle W\setminus A\rangle_\kappa)$, while the inequality in the other direction is obvious, therefore $|W\setminus A|=\dim(\langle W\setminus A\rangle_\kappa)$; similarly we have $|A\setminus W|=\dim(\langle A\setminus W\rangle_\kappa)$ and $|A\cap W|=\dim(\langle A\cap W\rangle_\kappa)$. Also $\dim(\hat{\kappa}\cap W' )=|A\cap W|$, as the second coordinates of $W\setminus A$ generate a subspace which does not intersect $\hat{\kappa}$ by definition of $\mathcal{P}_{chaff}$ and given that $B= \F_q^k\backslash A$.

Overall, it follows that $d_{\Delta}(A,W)= d_S(W', \hat{\kappa})$, and therefore  we can uniquely decode $W'$ to $\hat{\kappa}$ as soon as the set difference between $A$ and $W$ is at most $\frac{1}{2}(d_{\min,S}(\mathcal{C}) -1)$.
\qed
\end{proof}

\subsection{Variants of the scheme}

In order to loosen the constraints on the choice of parameters, other settings and scheme variants can be considered, although some probability of error may be introduced.

For example, we can allow $|A|=|W|=t\geq k$, with the features thought as randomly chosen in the ambient space rather than linearly independent. Other looser assumptions include also $B$ being a proper subset of $ \F_q^k\backslash A$.

In these cases, one needs to compare $\dim(\hat{\kappa}\cap W' )$ with $|A\cap W|$ and $\dim(\hat{\kappa}+ W' )$ with $|A\cup W|$. For example $\dim(\hat{\kappa}\cap W' )$ is no bigger than $k$ while $|A\cap W|$ would be no bigger than $t$;  $|A\cup W|$ counts elements of $A$ which do not contribute to the dimension of $\hat{\kappa}$; $\dim(W')$ may not be equal to $|W|$ and the looser assumption on $B$ may reduce the dimension of $W'$ even more, introducing further variability.

Depending on the assumptions and parameters, one can expect to have bounds of the form:
$$
d_{\Delta}(A,W)-\delta_1\leq  d_S(W', \hat{\kappa})\leq d_{\Delta}(A,W)+\delta_2,
$$
for some $\delta_1,\delta_2 \in \mathbb N$. Depending on the given threshold for $d_{\Delta}(A,W)$, one can estimate the probability of falsely accepting or falsely rejecting the witness.

To be more precise, with the above mentioned looser assumptions, we get $\dim(\hat{\kappa}) = k =|A|-(t-k)$, $\dim(W')\leq |W|$ and $\dim(\hat{\kappa} \cap W') \leq |A\cap W|$. If $y$ is an upper bound on the difference between $|A\cap W|$ and the maximum number of linearly independent elements within $A\cap W$ (i.e. $y=0$ for the hypothesis of Theorem~\ref{eq:sfvdistance}), we have on one side
\begin{align*} 
 d_S(W', \hat{\kappa}) &= \dim(\hat{\kappa})+\dim( W' )-2\dim(\hat{\kappa}\cap W' )\\
&\leq |A|-(t-k)+ |W|-2(|A\cap W|-y) \\ 
&= d_{\Delta}(A,W)- (t-k)+2y.
\end{align*} 
On the other side, if $z$ is an upper bound for $|W|-\dim(W')$, we get
\begin{align*} 
 d_S(W', \hat{\kappa}) &= \dim(\hat{\kappa})+\dim( W' )-2\dim(\hat{\kappa}\cap W' )\\
&\geq |A|-(t-k)+ |W|-z-2(|A\cap W|) \\ 
&= d_{\Delta}(A,W)- (t-k)-z .
\end{align*} 
Note that $z$ depends on the assumptions on the size of $B$ and on the choice of chaff points and the parameter $n$, as discussed in the first part of Section \ref{s:SFV}. I.e.\ $z$ can be neglected if $n$ is big enough, $B$ is the complement to $A$, and the chaff points are randomly chosen.
Similar bounds can also be obtained for $t<k$.

Incidentally, these inequalities provide an alternative proof to Theorem \ref{eq:sfvdistance}.


\section{Security and Examples}\label{s:secu}

Notice that we can use $n$ as a degree of freedom to enlarge the size of the key space. 

We know the following fact from \cite{Laksov:countingmatrices}: 

\begin{lemma}\label{l:spanspace}
	Let $k\leq \delta \leq n$. The number of $\delta\times n$ matrices over $\F_q$ with rank $k$ is given by \begin{equation}\label{eq:spanspace}N_q(k,\delta,n) =  \dfrac{\left(\prod_{i=0}^{k-1}q^n-q^i\right)\left(\prod_{i=0}^{k-1}q^\delta-q^i\right)}{\prod_{i=0}^{k-1}q^k-q^i}.\end{equation}
\end{lemma}

With $\delta=\kappa$ we can see that we can play on $n$ to make this number grow as we please, in order to make it hard searching for the right set of $k$ linearly independent features.


Moreover, the complexity of  such a brute force attack should be combined with the difficulty of determining the rank of an arbitrary $k\times n$ matrix over $\F_q$. The naive approach, using Gauss� algorithm, requires at most $n(k^2-k)$ field operations, and  in case the field is $\F_2$ at most $n(k^2-k)/2$. There exist fast algorithms for determining the rank of a matrix but these are only asymptotically better and are often much worse for small values of $k$ and $n$.

\subsection{Other attacks}

When $|A|=t>k$, not only may some difficulty in decoding arise, but if $t$ is much bigger than k, other types of brute force attacks may be devised.
In the following a strategy is described which tries to find a set in $\F_q^n$ containing $k$ linearly independent vectors that are meant to reveal the authentic features.

Assume now to have $t$ authentic points and $r-t$ chaff points, with the set of features $\{x_1,...,x_t\}$ being a set of random elements of $\F_q^k$. We can assume that the second coordinates of the authentic set $\{x_1\kappa, ..., x_t\kappa\}$ contain a set of $k$ linearly independent vectors in $\F_q^n$. Indeed, given Lemma \ref{l:spanspace}, we can compute the probability that $x_1\kappa, ... , x_t\kappa$ contains a set of $k$ linearly independent vectors as $$\frac{N_q(k,k,t)}{q^{kt}},$$
that is the probability that $(x_1, ..., x_t)^T$ is a rank $k$ matrix. For common vault parameters, and especially for larger $t$, this value is close to $1$, so as to justify our assumptions.






Now, the expected number of subsets of size $\delta$ out of $r>\delta$ random points in $\F_q^n$ that span a $k$-dimensional space can be estimated as \begin{equation}\label{eq:alpha} \alpha_q(k,\delta,n) = \dfrac{{r\choose \delta}N_q(k,\delta,n) }{ q^{\delta n}}.\end{equation} Ideally, an attacker would want to find a $\delta_0 \leq |A|=t$ 
 so that $\alpha_q(k,\delta_0,n)<1$ in order to have a high probability of recovering the key in the event that the $\delta_0$ points span a space of dimension $k$. On the other side, to counter this type of attack, one tries to keep $k$ very close to $t$ and $r$ big enough, so that $\alpha_q$ does not get small.

We will approximate the complexity of a brute force attack following this approach. The attack is similar in approach to that proposed in \cite{Juels:2006:FVS:1110940.1110956} and depends on finding a suitable $\delta_0$, so that the probability of $\delta_0$ random vectors in $\F_q^n$ spanning a subspace of dimension $k$ is small.



It is noted in \cite{DBLP:conf/biosig/MihailescuMT09} that the average number of attempts for a user to guess $\delta$ points in the authentic set is ${r\choose\delta}/{t\choose\delta} < 1.1(r/t)^\delta$ for $r>t>5$. Given that it takes $n(\delta^2-\delta)/2$ operations to row reduce a $\delta\times n$ binary matrix, we obtain the following upper bound for the expected time to recover the key.

\begin{lemma}\label{l:easybound}
	In the above settings, let $\delta_0$ be so that $\alpha_2(k,\delta_0,n) < 1$ from equation \eqref{eq:alpha}. On average, an attacker can recover the secret key in $C\cdot (r/t)^{\delta_0}$ operations, where $C < 0.55\cdot n(\delta_0^2-\delta_0)$. 
\end{lemma}

\subsection{Example using spread codes}\label{sec:example}


As an example of how to construct a vault using subspace codes, we will use spread codes, as defined in Section \ref{sec:background}. 

Spread codes are somewhat restrictive in that the minimum distance is completely determined by $k$, unlike other subspace codes where one can trade off the distance with other parameters. Nevertheless we illustrate the construction using spread codes because of their simplicity.

\begin{example}
Let us assume that the features belong to $\F_2^{16}$, so that $k=16$. In this case, we can recover the key if and only if the set difference is at most $15$. 
 We are free to choose $n$ as long as it is a positive integer multiple of $k$. For example we can choose $n = 96$ so that we have roughly $2^{80}$ keys. 

\end{example}

Note that an $(n,k)_q$ spread code can be decoded in $\mathcal{O}((n-k)k^5)$ field operations over $\F_{q}$, as shown in \cite{go12}. For more information on spread codes and other decoding algorithms, the reader is referred to \cite{go12,DBLP:conf/isit/ManganielloGR08,ma11j}.

\section{Further Considerations}\label{sec:considerations}

One of the disadvantages of using a biometric for security is that once an attacker knows a user's features, the user can never use a biometric scheme based on those features again. 
In the PFV finding the key is essentially equivalent to finding the features, as they are immediately retrievable as the first coordinates of the points in the authentic set, i.e. by testing whether these correspond to evaluations of the key polynomial. In the SFV, instead, an attacker who is capable of obtaining $\hat\kappa$, has no big advantage in recovering $x_1, ..., x_t$ from $x_1\kappa, ... ,x_t\kappa$, not knowing which particular $\kappa$ was used to generate the second coordinates of the authentic points. 
Ideally, to make the system even more resilient, the user should have the features obscured, for instance one might want to store in the vault a hash of the features, instead of the features themselves, as
\begin{align*} 	 
 \mathcal{P}_{auth} &= \{(h(x), x\kappa) \mid x \in A\} \\ 
	 \mathcal{P}_{chaff} &= \{(h(x), \lambda(x))\mid x\in B\},
\end{align*} 
for a suitable hash function $h$.
There is also another important reason to use hashes as above in the system. In fact, suppose that an attacker finds an element in the unhashed version of the vault whose first coordinate is a linear combination of other first coordinates of other elements in the vault. Then he can check whether its second coordinate is also a linear combination (with the same coefficients) of the corresponding second coordinates of the other elements. If this happens he can argue that the element belongs to $\mathcal{P}_{auth}$. Clearly also this attack can be prevented by taking $t$ close to $k$, besides using an hash function to hide the first coordinates.

\section{Conclusions}\label{sec:conclusions}

We have proposed a new authentication scheme based on noisy data like biometric features. The idea has similarities with the fuzzy vault scheme and works in the set difference metric, but it exploits the new setting of subspace codes. We have presented a main theorem with two alternative proofs that shows under which distance conditions authentication succeeds with respect to the code parameters. We have also showed the possibility of considering a few variants based on slightly different assumptions and how the main theorem can be generalized. This can allow more flexibility for the choice of parameters and for future applications.  The security of the scheme has been analyzed, whereby brute force attacks require bigger computational costs compared with traditional schemes. 
This however comes with a price, that is the computational complexity of state of the art decoding schemes for subspace codes is also rather high. There are also a few other nice features of the new scheme, for example its resilience to exposing the features even if the key were compromised.

Future research includes enhancing the scheme or devising alternative schemes based on subspace codes that would enable more efficient and flexible parameter profiles or decoding scenarios. Also considering examples with families of codes other than spread codes may help suggest future steps towards an actual deployment in practice.

\section{Acknowledgments}
The authors would like to thank Marco Bianchi and Natalia Silberstein for fruitful discussions regarding this work.

\bibliography{biblio}{}
\bibliographystyle{spbasic}

\end{document}